\documentclass{article}[11pt,letter]

\usepackage[T1]{fontenc}
\usepackage[english]{babel}
\usepackage[cmex10]{amsmath}
\usepackage[utf8]{inputenc}
\usepackage{amssymb}
\usepackage{amsthm}
\usepackage{mathtools}
\usepackage{xspace}
\usepackage[symbol*]{footmisc}
\usepackage[hyperfootnotes=false,breaklinks,bookmarks=false]{hyperref}
\usepackage{algorithmic}
\usepackage{algorithm}
\usepackage{color}
\usepackage{authblk}
\usepackage{fullpage}
\usepackage{cite}
\usepackage{todonotes}

\newtheorem{definition}{Definition}[section]
\newtheorem{lemma}[definition]{Lemma}
\newtheorem{theorem}[definition]{Theorem}

\newtheorem{corollary}[definition]{Corollary}

\newcommand{\bigo}{\mathcal{O}}

\newcommand{\symb}{R}

\newcommand{\etal}{{et~al.}\xspace}

\newcommand{\offset}{\mathsf{offset}}
\newcommand{\name}{\mathsf{name}}
\newcommand{\rank}{\mathsf{rank}}
\newcommand{\select}{\mathsf{select}}

\newcommand{\cost}{\delta}
\newcommand{\da}{\delta'}
\newcommand{\length}{\ell}
\newcommand{\encoding}{\mathsf{label}}

\begin{document}

\author[1]{Pawe\l{} Gawrychowski}
\affil[1]{Institute of Informatics, University of Warsaw, Poland}
\author[2]{Adrian Kosowski}
\affil[2]{Inria Paris and IRIF, Universit\'e Paris Diderot, France}
\author[3]{Przemys\l{}aw Uznański
}
\affil[3]{Department of Computer Science, ETH Z\"{u}rich, Switzerland}

\date{}

\title{Sublinear-Space Distance Labeling using Hubs}

\maketitle

\begin{abstract}
A distance labeling scheme is an assignment of bit-labels to the vertices of an undirected, unweighted graph such that the distance between any pair of vertices can be decoded solely from their labels. We propose a series of new labeling schemes within the framework of so-called hub labeling (HL, also known as landmark labeling or 2-hop-cover labeling), in which each node $u$ stores its distance to all nodes from an appropriately chosen set of hubs $S(u) \subseteq V$. For a queried pair of nodes $(u,v)$, the length of a shortest $u\!-\!v$-path passing through a hub node from $S(u)\cap S(v)$ is then used as an upper bound on the distance between $u$ and $v$.

We present a hub labeling which allows us to decode exact distances in sparse graphs using labels of size sublinear in the number of nodes.  For graphs with at most $n$ nodes and average degree $\Delta$, the tradeoff between label bit size $L$ and query decoding time $T$ for our approach is given by $L = \bigo(n \log \log_\Delta T / \log_\Delta T)$, for any $T \leq n$. Our simple approach is thus the first sublinear-space distance labeling for sparse graphs that simultaneously admits small decoding time (for constant $\Delta$, we can achieve any $T=\omega(1)$ while maintaining $L=o(n)$), and it also provides an improvement in terms of label size with respect to previous slower approaches.

By using similar techniques, we then present a $2$-additive labeling scheme for general graphs, i.e., one in which the decoder provides a 2-additive-approximation of the distance between any pair of nodes. We achieve almost the same label size-time tradeoff $L = \bigo(n \log^2 \log T / \log T)$, for any $T \leq n$. To our knowledge, this is the first additive scheme with constant absolute error to use labels of sublinear size. The corresponding decoding time is then small (any $T=\omega(1)$ is sufficient).

We believe all of our techniques are of independent value and provide a desirable simplification of previous approaches.
\end{abstract}

\thispagestyle{empty}

\newpage

\section{Introduction}
Distance labeling schemes, popularized by Gavoille et al.~\cite{Gavoille:2004:DLG:1036161.1036165}, are among the most fundamental distributed data structures for graph data. The design problem combines two major challenges. First of all, distance labelings serve the role of a \emph{distance oracle}, i.e., a data structure which for a given undirected graph $G=(V,E)$ can answer queries of the form: ``what is the distance between the nodes $s, t \in V$?''. Throughout most of this paper, we will assume that $G$ is an unweighted graph with $n$ nodes and $m$ edges. The efficiency of a distance oracle is measured by the interplay between the \emph{space} requirement of the data structure representation, the \emph{encoding time} required to set up the oracle for a given graph, and perhaps more importantly, its \emph{decoding time}, that is, the time of processing a $s-t$ distance query. Moreover, a distance labeling scheme is defined more restrictively than a distance oracle, as an assignment of a binary string (label) $\encoding(u)$ to each node $u\in V$, so that the graph distance between $u$ and $v$ is uniquely determined by the pair of labels: $\encoding(u)$ and $\encoding(v)$. The \emph{size} of a distance labeling scheme is now the maximum length of a node label in the graph. In this way, distance labelings add an extra layer of complexity to the graph distance decoding problem, by imposing a distributed representation of information in the labels $(\encoding(u) : u\in V)$. Whereas the concatenation of all $n$ labels in a distance labeling forms a centralized distance oracle, distance labelings can also be applied in a distributed setting, in which the label of each node is stored at a distinct location in the network. This is the case, for instance, in applications in compact routing protocols, where the goal is to find a shortest path from a source node to a target node with a known label~\cite{ChepoiDEHVX12}.

An interesting characteristic of the problem of distance oracle design for sparse graph is its inherent link to an underlying set intersection task. On the side of lower bounds, this is most clearly observed, following Pătraşcu and Roditty~\cite{PatrascuR14}, when we consider a pair of vertices belonging to the same partition of a bipartite graph. The distance between them is $2$ if and only if the sets of their neighbors intersect, and at least $4$ otherwise. Consequently, assuming a plausible conjecture on the space required to decide intersection of a set of small sets, it follows that any oracle for graphs with $\widetilde \bigo(1)$ maximum degree, which admits constant decoding time, requires $\widetilde \Omega (n^2)$ space. (Here, the $\widetilde O$ and $\widetilde \Omega$ notation disregards polylogarithmic factors in $n$.) By contrast, many efficient algorithms for answering distance queries in real-world scenarios rely on the premise that the distance between a pair of nodes can be computed using an intersection-type query on a pair of small sets. In the basic framework of hub labelings, see \cite{Abraham:2012:HHL:2404160.2404164}, (introduced in~\cite{Cohen:2003:RDQ:942270.944300} under the name of 2-hop covers, and also referred to as landmark labelings~\cite{Abraham11onapproximate}), each node $u\in U$ stores the set of its distances to some subset $S(u) \subseteq V$ of other nodes of the graph. Then, the computed distance value $\da(u,v)$ for a queried pair of nodes $u, v\in V$ is returned as:
\begin{equation}\label{eq:distance}
\da(u,v) := \min_{w \in S(u)\cap S(v)} \cost(u,w) + \cost(w,v),
\end{equation}
where $\cost$ denotes the shortest path distance function between a pair of nodes. The computed distance between all pairs of nodes $u$ and $v$ is exact if set $S(u)\cap S(v)$ contains at least one node on some shortest $u-v$ path. This property of the family of sets $(S(u) : u\in V)$ is known as \emph{shortest path cover}. The hub-based method of distance computation is in practice effective for two reasons. First of all, for transportation-type networks it is possible to show bounds on the sizes of sets $S$, which follow from the network structure. Notably, Abraham et al.~\cite{doi:10.1137/1.9781611973075.64} introduce the notion of highway dimension $h$ of a network, which is presumed to be a small constant e.g.\ for road networks, and show that an appropriate cover of all shortest paths in the graph can be achieved using sets $S$ of size $\widetilde \bigo(h)$. Moreover, the order in which elements of sets $S(u)$ and $S(v)$ is browsed when performing the minimum operation is relevant, and in some schemes, the operation can be interrupted once it is certain that the minimum has been found, before probing all elements of the set. This is the principle of numerous heuristics for the exact shortest-path problem, such as contraction hierarchies and algorithms with arc flags~\cite{Kohler06fastpoint-to-point,Bauer:2010:SFR:1498698.1537599}.

In this work, we make use of the hub set techniques to obtain better (distributed) distance labelings. Whereas $\Omega(n)$ is a lower bound of the size of a hub set for general graphs, we provide hub-based schemes using smaller sets for specific case, leading to labels which can be encoded on $o(n)$ bits. Our scheme provides a shortest-path cover in the class of sparse graphs (with average degree $\Delta=2m/n$ subpolynomial in $n$). This construction is overviewed in more detail in Section~\ref{sec:our_results}.

The implications of our result can be seen as twofold. First of all, our approach directly leads to labeling of smaller size (and smaller decoding time) for exact distance queries in sparse graphs than all previous distance labeling approaches. Additionally, an corollary of our result concerns the case of \emph{$k$-additive} approximate distance labeling, in which the distance decoder is required to return an upper bound on the shortest path length which is within an additive factor of at most $k$ from the optimum. So far, no way to construct $k$-additive distance labels using labels of sublinear size in $n$ was known for any constant $k>0$. (This question was considered previously in, e.g.,~\cite{DBLP:conf/soda/AlstrupGHP16}). We provide a way to construct a $2$-additive distance labeling in general graphs using distance labels of size $o(n)$. (This result is essentially the best possible, since a $1$-additive distance labeling requires distance labels of size at least $n/4$ already on the class of bipartite graphs.)

In our approaches, the size of the obtained distance labels for the considered cases is improved with respect to the state-of-the-art by up to a logarithmic multiplicative factor. Rather than seeing this as ``gaining'' a logarithm, we rather see this as ``not losing'' a logarithm. Indeed, the basic ingredient of the hub sets in previous approaches was a subset of nodes, sampled independently at random from $V$~\cite{BCE05,Sublinear}. The constructions then relied on the probabilistic method to guarantee that the hubs would have the shortest-path cover properties, based on the premise that for each pair of nodes the constructed hubs provide a shortest path cover with sufficiently high probability. The derandomization of this process resulted in a loss of a logarithmic factor in the analysis of the size of labels. Our approach shows how to avoid this issue: when constructing labelings for sparse graphs, we do away with randomization altogether, relying on simple structural results to replace the random subset of nodes.

\subsection{Related Work}

\paragraph{Distance Labelings.}
The distance labeling problem in undirected graphs was first investigated by Graham and Pollak~\cite{pollak}, who provided the first labeling scheme with labels of size $\bigo(n)$. The decoding time for labels of size $\bigo(n)$ was subsequently improved to $\bigo(\log \log n)$ by Gavoille \etal~\cite{Gavoille:2004:DLG:1036161.1036165} and to $\bigo (\log^* n)$ by Weimann and Peleg~\cite{WP11}. Finally, Alstrup \etal~\cite{DBLP:conf/soda/AlstrupGHP16} present a scheme for general graphs with decoding in $\bigo(1)$ time using labels of size $\frac{\log 3}{2} n + o(n)$ bits.\footnote{For the sake of sanity of the notation, we define $\log x = \max(1, \log_2(x))$.} This matches up to low order terms the space of the currently best known distance oracle with $\bigo(1)$ time and $\frac{\log 3}{2} n^2 + o(n^2)$ total space in a \emph{centralized} memory model, due to Nitto and Venturini~\cite{NV08}.

The notion of $D$-preserving distance labeling, first introduced by Bollob\'as \etal~\cite{BCE05}, describes a labeling scheme correctly encoding every distance that is at least $D$. \cite{BCE05} presents such a $D$-preserving scheme of size $\bigo(\frac{n}{D} \log^2 n)$.
This was recently improved by Alstrup \etal~\cite{Sublinear} to a $D$-preserving scheme of size $\bigo(\frac{n}{D}\log^2 D)$. Together with an observation that all distances smaller than $D$ can be stored directly, this results in a labeling scheme of size $\bigo(\frac{n}{\symb}\log^2 \symb)$, where $\symb = \frac{\log n}{\log \frac{m+n}{n}}$. For sparse graphs, this is $o(n)$.

For specific classes of graphs, Gavoille \etal~\cite{Gavoille:2004:DLG:1036161.1036165} described a $\bigo(\sqrt{n}\log n)$ distance labeling for planar graphs, together with $\Omega(n^{1/3})$ lower bound for the same class of graphs. Additionally, $\bigo(\log^2 n)$ upper bound for trees and $\Omega(\sqrt{n})$ lower bound for sparse graphs were given.

\paragraph{Distance Labeling with Hub Sets.} For a given graph $G$, the computational task of minimizing the sizes of hub sets $(S(u) : u\in V)$ for exact distance decoding is relatively well understood. A $\bigo(\log n)$-approximation algorithm for minimizing the average size of a hub set having the sought shortest path cover property was presented in Cohen \etal~\cite{Cohen:2003:RDQ:942270.944300}, whereas a $\bigo(\log n)$-approximation for minimizing the largest hub set at a node was given more recently in Babenko \etal~\cite{DBLP:conf/icalp/BabenkoGGN13}. Rather surprisingly, the structural question of obtaining bounds on the size of such hub sets for specific graph classes is wide open. For example, for the class of graphs of constant maximum degree, there is a large gap between the hub sets in our construction (of size $\bigo(n/\log n)$) and the generic lower bound of $\widetilde \Omega(\sqrt n)$.


\paragraph{Distance Oracles.}

A centralized version of distance labeling problem is \emph{distance oracle} problem, where one asks for a centralized data structure allowing for querying a distance between pair of vertices. There usually one asks for what type of tradeoffs are possible between size of the structure, time of the query and allowed error (multiplicative stretch). Sommer~\etal \cite{sommer2009distance} proved that any constant time, constant stretch oracle must be superlinear in $n$. Thorup and Zwick \cite{Thorup:2001:CRS:378580.378581} proved that distance oracles of stretch 2 require $\Omega(n^2)$ space, and of stretch 3 require $\Omega(n^{3/2})$ space. Pătraşcu and Roditty \cite{PatrascuR14} strengthened the lower bound for stretch 2, proving a lower bound of $\Omega(n\sqrt{m})$ on the size of oracles with constant query time. For general weighted graphs, Thorup and Zwick \cite{Thorup:2001:CRS:378580.378581} designed a distance oracle of size $\bigo(kn^{1+1/k})$, stretch-$(2k-1)$ and $\bigo(k)$ time. The query time has been improved to $\bigo(\log k)$ time by Wulff-Nilsen \cite{doi:10.1137/1.9781611973105.39}, and to constant time in Chechik \cite{Chechik:2014:ADO:2591796.2591801}. The size of the distance oracle from \cite{Thorup:2001:CRS:378580.378581} is optimal assuming girth-conjecture. For sparse graphs, \cite{PatrascuR14} design distance oracle of size $\bigo(n^{4/3}m^{1/3})$ and stretch 2 in constant time. Also in \cite{PatrascuR14}, a conditional lower bound of $\widetilde\Omega(n^2)$ bits for a constant time distance oracle is provided. Cohen and Porat \cite{DBLP:journals/corr/abs-1006-1117} extended this result to sparse graphs.  An up-to-date survey of results on approximate distance oracles is provided in~\cite{DBLP:reference/algo/Roditty15a}.

\subsection{Our Results and Organization of the Paper}\label{sec:our_results}

We start by introducing the necessary conventions in Section~\ref{sec:preliminaries}. We also describe the basic building block for encoding distance labels, namely, an efficient method of storing the hub set of a node, together with corresponding distances, in its distance label.

In Section~\ref{sec:sparse}, we show how to construct an exact distance labeling scheme for graphs of bounded maximum degree. This relies on hub sets which consist, for a vertex $u$ of the union of all nodes from a small ball around vertex $u$, and all nodes from a selection of equally-spaced levels of the breadth-first-search tree of $u$. We then apply a trick, known from the previous work of \cite{Stretch}, to reduce the problem of constructing a labeling scheme for a graph with bounded average degree to that of constructing a labeling scheme for a~bounded-degree graph on twice as many nodes. For graphs with at most $n$ nodes and average degree $\Delta$, the tradeoff between label bit size $L$ and query decoding time $T$ for our approach is given by $L = \bigo(n \log \log_\Delta T / \log_\Delta T)$, for any $T \leq n$.
In particular, setting $T=n$, we obtain labels of size $\bigo(\frac{n}{R} \log R)$, which improves previously best result \cite{Sublinear} by a factor of $\log R$, keeping the $\widetilde\bigo(n)$ decoding time. On the other end, setting $T=\log n$ we obtain first sublinear size distance labeling that achieves almost-constant decoding time.

In Section~\ref{sec:general}, we adapt our approach to general graphs, using a variant of the proposed labeling scheme for sparse graphs to achieve 2-additive approximation of distances. As before, we achieve a tradeoff between label size $L$ and time $T$ of the form $L = \bigo(n \log^2 \log T / \log T)$, for any $T \leq n$. This 2-additive distance labeling scheme can be easily transformed into an exact one, by encoding the difference between the estimation and the true distances. Since this difference is always from $\{0,1,2\}$, we achieve labels of size $\frac{\log 3}2 n + o(n)$ (with any $\omega(1)$ decoding time), or of size $(\frac {\log 3}2 + \varepsilon) n$ (with $\bigo(1)$ decoding time), for any $\varepsilon>0$. Our approach almost matches the size of the best known distance labeling schemes~\cite{DBLP:conf/soda/AlstrupGHP16}, which make use of labels of size $\frac{ \log 3}{2} n + o(n)$ to achieve $\bigo(1)$ decoding time. Arguably, our approach may be considered simpler.

We remark that all our results apply to unweighted graphs, in which each edge has unit length. For sparse graphs, in which each edge has an integer weight from some interval $[1,W]$, we can use the same hub sets with an appropriately modified encoding to achieve a time-label tradeoff of $L = \bigo(n \log \log_\Delta T \log W/ \log_\Delta T)$. For the additive scheme, by subdividing each edge of length $w \in [1,W]$ into a chain of unweighted edges (of length 1), we achieve a conversion of the 2-additive distance labeling scheme into a $(2W)$-additive-distance scheme for weighted graphs.

\section{Preliminaries}\label{sec:preliminaries}

\paragraph{Notation and Conventions.}
Even though we are mainly interested in unweighted graphs, for technical reasons in Sections~\ref{sec:sparse} and~\ref{sec:general} we will work in a more general setting where every edge of a graph has a fixed cost from the set $\{0,1\}$. $\cost(u,v)$ denotes the cost of a cheapest path
connecting a pair of nodes $u$ and $v$, and $\length(u,v)$ denotes the smallest number of edges on such a~path. We will require the constructed distance labeling to return the value of $\cost(u,v)$. The degree of a node $v$ is denoted by $\deg(v)$. When analyzing the complexity of the decoding, we assume standard word RAM with logarithmic word size, where we are allowed to access $\log n$ consecutive bits of the stored binary string in constant time.

From now on, we assume that the graph is connected. This is enough because we can always include the identifier of its connected component in the label of every node, and return $\infty$ if $u$ and $v$ belong to different connected components; this only induces additive $\log n$ overhead to the label size.

\paragraph{Encoding Distances and Identifiers.}

The basic procedure for encoding a hub set in a label exploits some ideas from~\cite{DBLP:conf/soda/AlstrupGHP16}; we provide a self-contained exposition for completeness. We fix an arbitrary spanning tree of the graph and assign preorder numbers in the tree to the nodes, i.e., node numbered $1$ corresponds to the root and so on. The preorder number of a node $u$ is denoted by $\name(u)$. Such a~numbering has the following useful property.

\begin{lemma}
\label{lem:sum}
Let $v_1,v_2,\ldots,v_n$ be the preorder sequence of all nodes.
Then, for any node $u$, $\sum_{i=2}^{n}|\cost(u,v_{i-1})-\cost(u,v_i)|\leq 2n$.
\end{lemma}

\begin{proof}
Consider an Euler tour corresponding a traversal of the chosen spanning tree. Every node is visited at least
once there, and the total length of the tour is at most $2n$. Consequently, we can cut the tour
into paths connecting node $v_{i-1}$ with node $v_i$, for every $i=2,3,\ldots,n$. The total length
of all these paths is at most $2n$ and the claim follows.
\end{proof}

The following lemma is used for encoding a hub set $S$ using $\bigo(|S| \log (n/|S|))$ bits.

\begin{lemma}
\label{lem:set_encoding}
For a fixed $v$ and set $S$ such that $|S| \le \frac{n}{x}$, set $S$ and all of the distances $\cost(v,u)$ for $u \in S$ can be stored in $\bigo(\frac{n}{x} \log x)$ bits. For any constant $t>0$, the representation
can be augmented with $\bigo(\frac{n}{\log^t n})$ additional bits so that all elements of $S$ can
be extracted one-by-one in $\bigo(|S|)$ total time and given any $u$ we can check if $u\in S$
(and if so, extract $\cost(u,v)$) in $\bigo(1)$ time.
\end{lemma}

\begin{proof}
Let $S = (v_1,\ldots,v_{|S|})$, where $\name(v_1) < \name(v_2) < \ldots < \name(v_{|S|})$. We store $\name(v_1)$ and then the differences $\name(v_2) - \name(v_1), \ldots, \name(v_{|S|}) - \name(v_{|S|-1})$. Every difference is encoded using the Elias $\gamma$ code (see Elias~\cite{EliasGamma}), and the encodings are concatenated to form one binary
string. We are storing up to $\frac{n}{x}$ integers whose absolute values sum up to at most $n$, so by Jensen's inequality this takes $\bigo(\frac{n}{x} \log x)$ bits in total. Similarly, we store $\cost(u,v_1)$ and then the differences $\cost(u,v_2) - \cost(u,v_1), \ldots, \cost(u,v_{|S|}) - \cost(u,v_{|S|-1})$. By Lemma~\ref{lem:sum} we are again storing up to $\frac{n}{x}$ numbers whose absolute values sum up to at most $2n$, which takes $\bigo(\frac{n}{x} \log x)$ bits.

All $v_i$ can be extracted one-by-one in $\bigo(1)$ time each with standard bitwise operations.
To facilitate checking if $x\in S$ in $\bigo(1)$ time, we observe that it is enough to store a bit-vector
$B[1..n]$, where the $\name(v_i)$-th bit is set to {\bf 1}, for every $i=1,2,\ldots,|S|$. Then checking if
$x\in S$ reduces to two $\rank_{\bf 1}$ queries. A $\rank_{\bf 1}$ query counts {\bf 1}s in the specified prefix of the bit-vector and a $\select_{\bf 1}$ query returns the position of the $k$-th {\bf 1} in the bit-vector. By
the result of P\v{a}tra\c{s}cu~\cite{Succincter}, for any constant $t>0$, a bit-vector of length $n$ containing $\frac{n}{x}$
{\bf 1}s can be stored using
\[\log{ n \choose \frac{n}{x}}+\bigo(\frac{n}{\log^t n})=\bigo(\frac{n}{x}\log x)+\bigo(\frac{n}{\log^t n})\]
bits so that any rank or select query can be answered in $\bigo(t)$ time.
This allows us to check if $u\in S$ and calculate $i$ such that $u=v_i$ in $\bigo(t)$ time. To retrieve
$\cost(u,v_i)$, we store two additional bit-vectors $B_+$ and $B_-$. Each of them contains exactly
$\frac{n}{x}$ {\bf 1}s and up to $2n$ {\bf 0}s. The bit-vectors are defined as follows. For each $i=2,3,\ldots,n$
we consider the difference $\eta=\cost(u,v_{i})-\cost(u,v_{i-1})$. If $\eta\geq 0$, we append
${\bf 0}^\eta {\bf 1}$ to $B_+$ and ${\bf 1}$ to $B_-$. Otherwise, we append
${\bf 1}$ to $B_+$ and ${\bf 0}^{-\eta}{\bf 1}$ to $B_-$. By Lemma~\ref{lem:sum}, each of these
two bit-vectors contains at most $2n$ {\bf 0}s, so they can be stored using
$\bigo(\frac{n}{x}\log x+\frac{n}{\log^t n})$ bits so that any rank or select query can be answered
in $\bigo(t)$ time. To recover $\cost(u,v_i)$, we need to sum up all the differences. This reduces
to summing up all positive and all negative differences separately, which can be done using the
corresponding bit-vector with one $\rank_{\bf 1}$ and one $\select_{\bf 0}$ query in $\bigo(t)$ total time.
\end{proof}

We remark that the above encoding and decoding scheme is efficient for sets of size $|S| = \widetilde \bigo(n)$. For smaller sets, we will simply use an explicit encoding of all distances in $S$, requiring $\bigo(|S|\log n)$ bits.

\section{Exact Distance Labeling in Sparse Graphs}\label{sec:sparse}

\subsection{Graphs of Bounded Maximum Degree}\label{sec:bdg}

In this subsection, we assume that $\deg(u)\leq \Delta$ for every node $u$. We consider distance labeling schemes characterized by a time parameter $T$. Intuitively, in the construction, $\symb=\frac{\log T}{\log\Delta}$ will be a threshold parameter, distinguishing small distances from large distances in the graph --- a node will be able to afford to explicitly store the distances and identifiers of all nodes up to some distance $\bigo(\symb)$ from itself in its distance label. Although this case is of independent interest, we are considering it as a building block for construction of labeling in graphs of bounded average degree. Thus graphs considered here are weighted with edge weights from $\{0,1\}$, for the reason explained in Section~\ref{sec:bounded_average}.

The rest of this subsection is devoted to the proof of the following Theorem.

\begin{theorem}\label{thm:bounded}
Fix any value $\Delta \le T \le n$  and let $\symb=\log_{\Delta} T$. In bounded-degree graphs, there is a labeling scheme of size $\bigo(\frac{n}{\symb} \log \symb)$ and decoding time $\bigo(T)$.
\end{theorem}

Let us denote $\symb' = \lfloor \symb \rfloor$. Since $\symb' \ge 1$, we can bound $\symb \ge \symb' \ge \frac12 \symb$.
Consider a node $u$. The ball of radius $r$ centered at $u$, denoted $B_u(r)$, is the set of nodes
which can be reached from $u$ by following at most $r$ edges. Because the degrees of all nodes
are bounded by $\Delta$, $|B_u(r)| = \bigo(\Delta^{r})$. The $k$-th layer centered at $u$,
denoted $L_u(k)$, consists of all nodes $v$ such that $\length(u,v) = k \pmod \symb'$. Because
the layers are disjoint, there exists an $\offset(u)\in \{0,1,\ldots,\symb'-1\}$ such that
$|L_u(\offset(u))|\leq \frac{n}{\symb'}$.

\paragraph{Definition of the Labeling.} We define the hub set of node $u$, to which it stores all its distances, as $S(u) := B_u(\symb') \cup L_u(\offset(u))$, see Fig.~\ref{fig:layers}. Formally, the label of $u$ consists of the following:
\begin{enumerate}
\item $n$ and $\name(u)$,
\item $\name(v)$ and $\cost(u,v)$ for every $v\in B_u(\symb')$,
\item $\name(v)$ and $\cost(u,v)$ for every $v\in L_u(\offset(u))$.
\end{enumerate}

\begin{figure}[t]
  \centering
    \vspace*{-2mm}
    \includegraphics[scale=1]{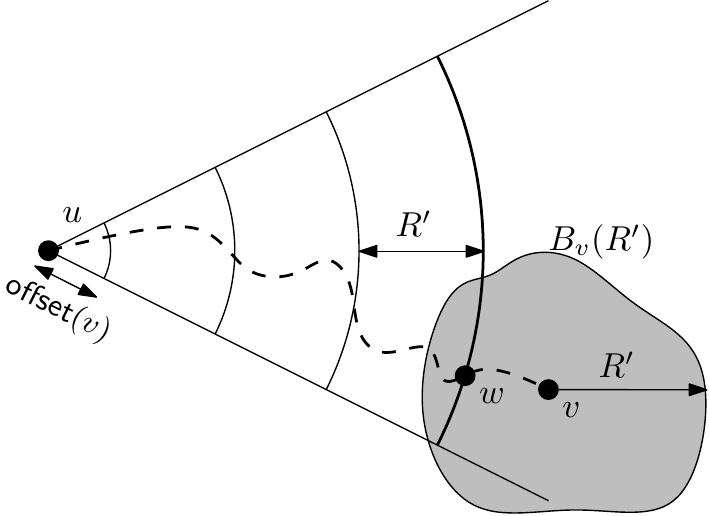}
    \vspace*{-2mm}
    \caption{Shortest path from $u$ to $v$ goes through $w$ which belongs to both $B_v(\symb')$ and $L_u(\offset(u))$.}
    \label{fig:layers}
\end{figure}

\paragraph{Computing $\cost(u,v)$.} For reasons of efficiency, we will not perform the distance decoding following Eq.~\eqref{eq:distance} directly, but we will treat the two components of the hub set of each node separately. Given $\name(u)$ and $\name(v)$, we can determine
$\cost(u,v)$ as follows. First we check if $v\in B_u(\symb')$ and if so return the stored $\cost(u,v)$.
Otherwise, we iterate through all nodes $w\in B_u(\symb')$ and check if $w \in L_v(\offset(v))$.
If so, we know $\cost(u,w)+\cost(w,v)$. We return the smallest such sum.

For the proof of correctness of the distance decoder, it is clear that $\cost(u,w)+\cost(w,v)\geq \cost(u,v)$ for any $w$, so it remains to argue that either $v \in B_u(\symb')$ or there exists $w\in B_u(\symb')$ such that $w\in L_v(\offset(v))$ and
$\cost(u,w)+\cost(w,v)=\cost(u,v)$. Consider a shortest path $(p_0, p_1, p_2, \ldots p_\length)$ where $v=p_0$ and $u=p_\length$ such that $\length = \length(u,v)$. If  $\length \leq \symb'$, $v\in B_u(\symb')$ and there is nothing to prove,
so we can assume that $\length > \symb'$.  Observe that for any $i=0,1,\ldots,\length$, $\length(v,p_i)=i$,
so in particular $p_{\alpha\cdot \symb'+\offset(v)}\in L_v(\offset(v))$ for any integer $\alpha \geq 0$.
We choose $\alpha=\left\lfloor \frac{\length-\offset(v)}{\symb'}\right\rfloor$ and $w=p_{\alpha \cdot \symb'+\offset(v)}$.
Then $w\in L_v(\offset(v))$, $w\in B_u(\symb')$ by the choice of $\alpha$, and $\cost(u,w)+\cost(w,v) = \cost(u,v)$
because $w$ lies on a shortest path connecting $u$ and $v$, so indeed we are able to correctly
determine $\cost(u,v)$.

\paragraph{Encoding and Size of the Scheme.} Encoding $n$ and $\name(u)$ takes $\bigo(\log n)$ bits.
The set $B_u(\symb')$ with corresponding distances is stored explicitly, while set $L_u(\offset(u))$ together with the corresponding distances is stored using Lemma~\ref{lem:set_encoding}, using $\bigo(\Delta^{R'}\log n) = \bigo(T \log n)$ and
$\bigo(\frac{n}{\symb'} \log \symb')$ bits, respectively.
Hence the total size of the scheme is 
\[\bigo(\log n+T\log n+\frac{n}{\symb'}\log \symb')=\bigo(\frac{n}{\symb}\log \symb ),\]
where we have used the fact that for any $T =  \text{poly}(n)$ the claimed label size is the same, thus we can assume $T = o(n/\text{polylog}(n))$.

\paragraph{Complexity of the Decoding.} Checking if $v\in B_u(\symb')$ and retrieving the encoded $\cost(u,v)$ takes $\bigo(T)$ time.
Similarly,  iterating through all $w\in B_u(\symb')$, checking if $w\in L_v(\offset(v))$ and
if so retrieving the encoded $\cost(v,w)$ takes, by Lemma~\ref{lem:set_encoding}, $\bigo(1)$ time per single $w$,  thus $\bigo(|B_u(\symb')|)=\bigo(T)$ total time.
All in all, we can compute $\cost(u,v)$ in $\bigo(T)$
total time.
\qed

\paragraph{Smaller values of $T$.} For the sake of completeness, we consider the special case of $T < \Delta$. Consider labeling where the label of a node
$u$ consists of $n$, $\name(u)$, and all values $\cost(u,v)$ for $v \in V$ stored using
Lemma~\ref{lem:set_encoding}. This takes $\bigo(n)$ bits, with $\bigo(1)$ decoding time, and matches claimed bounds from Theorem~\ref{thm:bounded}.

\ \\
We also observe that our result applies not only to distance labels, but also as a size upper bound of hub sets for sparse graphs. Indeed, by fixing $T = n$, and observing that $|B_u(R')| + |L_u(\offset(u))| \le \frac{n}{R'}$, we have the following:
\begin{corollary}
In bounded-degree graphs, there is a hub set construction of size $\bigo(\frac{n}{\log_{\Delta} n})$ vertices per node.
\end{corollary}

\subsection{Graphs of Bounded Average Degree}
\label{sec:bounded_average}

We now allow for bounded average degree by reduction to the approach from Subsection~\ref{sec:bdg}. Given a graph $G$, let $\Delta=\frac{m+n}{n}$. We will create a new graph by splitting nodes of high degree. Following the formulation from~\cite[Lemma 4.2]{Stretch} (cf. Figure \ref{fig:splitting}), we can obtain a graph $G'$ on at most $2n$ nodes and at most $m+n$ edges, such that the degree of every node is bounded by $\left\lceil \frac{m}{n} \right\rceil+2 \le \Delta + 2$ and the distance between two nodes in the original graph $G$ is exactly the same as the distance between their corresponding nodes in the new graph $G'$.
\begin{figure}[t]
  \centering
    \vspace*{-2mm}
    \includegraphics[scale=1.8]{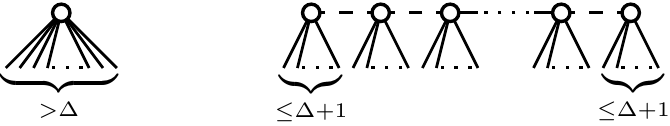}
    \vspace*{-2mm}
    \caption{Example of subdividing of a large degree node (on the left) into a family of nodes of small degree, connected by edges of weight 0 (dashed edges).}
    \label{fig:splitting}
\end{figure}
We can now directly apply the scheme from Theorem~\ref{thm:bounded} to graph $G'$, and exactly the same distance labels will work for the corresponding nodes of graph $G$. In this way, we obtain a scheme of size:
\[\bigo(\frac{n}{\symb}\log \symb + \frac{m}{n}) = \bigo(\frac{n}{\symb}\log \symb), \text{ where } \symb=\bigo(\frac{\log T}{\log\Delta}).\]
which returns $\cost(u,v)$ in $\bigo(T)$ time given the labels of $u$ and $v$. The correctness of the this reduction is guaranteed by the fact that Theorem~\ref{thm:bounded} allows for edge weights from $\{0,1\}$.

\begin{theorem}
Fix any  $T \le n$ and $\Delta$, and let $R = \frac{\log T}{\log \Delta}$. There exists an exact distance labeling for graphs with average degree $\Delta$ using labels of size $\bigo( \frac{n}{R} \log R)$
 and a corresponding decoding scheme requiring time $\bigo(T)$.\qed
\end{theorem}

It is easy to see that this reduction preserves bounds on the size of hub sets, so we have the following:
\begin{corollary}
In graphs with average degree $\Delta$, there is a hub set construction of size $\bigo(\frac{n}{\log_{\Delta} n})$ vertices per node.
\end{corollary}

\section{2-Additive Distance Labeling in General Graphs}\label{sec:general}

We will apply a similar distance labeling scheme as for sparse graphs, obtaining a 2-additive approximation of the distance between any pair of  with label sizes of $o(n)$ per node. In this approximate scheme, the hub sets will have the following property.  The label of each node $ u \in V$ will provide an encoding of the node identifiers of a subset $S(u) \subseteq V$ and of the distances from $u$ to all elements of $S(u)$. The sets $S(u)$ will be defined so that for any pair $u, v$, there exists a node $w \in S(u) \cap S(v)$, such that either $w$ or a neighbor of $w$ lies on the shortest path from $u$ to $v$ in $G$. We will decode the approximate distance as before, using Eq.~\eqref{eq:distance}; clearly, $\da (u, v) \in \cost (u, v) + \{0,1,2\}$.

The construction of sets $S(u)$ is performed as follows. Let $\tau < \frac{1}{2} \log n$ be an threshold value of vertex degree, to be chosen later. Let $V' = \{ v\in V : \deg(v)>\tau\}$, and let $S'\subseteq V$ be a minimal dominating set for $V'$, i.e., a subset of $V$ with the property: $\forall_{w\in V'} B_w(1) \cap S'\neq\emptyset$. By a straightforward application of the probabilistic method~(cf.\cite[proof of Theorem 1.2.2]{probmethod}), we have that there is $S'$ such that $|S'| \leq \frac{1+\ln(\tau+1)}{\tau+1} n < \frac{2\ln \tau}{\tau} n$, and it can be easily constructed in polynomial time (a deterministic construction by a folklore greedy algorithm gives set of size $\bigo(\frac{\ln \tau}{\tau} n)$). For every $u\in V$, we define $B'_u(r)$ as the set of nodes of the ball of radius $r$ around $u$ in the subgraph $G[V\setminus V']$.  Finally, we define $\symb = \frac{\tau}{\log \tau}$ and let $L_u$, $\symb'$, and $\offset(u)$ be defined as in Section~\ref{sec:bdg}, and let $S'_u$ be a minimal subset of $S'$ such that for every $w\in V'$ adjacent the boundary of $B'_u(\symb')$, i.e. $B_w(1) \cap B'_u(\symb') \neq \emptyset$, we have $B_w(1) \cap S'_u \neq \emptyset$. Such $S'_u$ can be easily constructed in polynomial time, and moreover, since there are at most $\tau^{R'+1}$ vertices adjacent to the boundary, we have $|S'_u| = 2^{\bigo(\tau)}$.

The approximate distance label of $u$ now consists of the following elements:
\begin{enumerate}
\item $n$ and $\name(u)$,
\item $\name(v)$ and $\cost(u,v)$ for every $v\in B'_u(\symb')$,
\item $\name(v)$ and $\cost(u,v)$ for every $v\in L_u(\offset(u))$.
\item $\name(v)$ and $\cost(u,v)$ for every $v\in S'_u$,
\item $\name(v)$ and $\cost(u,v)$ for every $v\in S' \setminus S'_u$.
\end{enumerate}
The separation of $S'$ into $S'_u$ and $S' \setminus S'_u$ in the label is done to allow
efficient decoding.

\paragraph{Computing $\cost(u,v)$.}
To show the correctness of this approximate labeling scheme, fix a pair of vertices $u, v\in V$. If there exists a vertex $w\in B'_u(\symb')$ lying on a fixed shortest path $P$ between $u$ and $v$ such that $w=v$ or $w \in L_v(\offset(v))$, then the labeling scheme finds the shortest path distance between $u$ and $v$ as in Section~\ref{sec:bdg}. Otherwise, let $y$ be the nearest vertex to $u$ lying in $P \setminus B'_u(\symb')$; it follows from the construction that $y \in V'$. Then, there exists $w\in B_y(1)$ such that $w\in S'_u \subseteq S'$. In this case, the distance $\cost(u,w) + \cost (v,w)$ is a $2$-additive approximation of $\cost(u,v)$.

\paragraph{Size of the Scheme.}
The size of the label of a node $u$ in the scheme can be bounded as follows: $|B'_u(\symb')| \leq \tau^{\symb'}\leq 2^\tau$, $|L_u(\offset(u))| \leq \frac n\symb$, $S'  < \frac{2\ln \tau}{\tau} n$. Overall, the total size is $|B'_u(\symb')\cup L_u(\offset(u)) \cup S'| = \bigo(\frac{\log \tau}{\tau} n)$, thus using Lemma~\ref{lem:set_encoding} to store the sets and the corresponding distances we obtain labels of
size $\bigo(n \frac{\log^2 \tau}{\tau})$.

\paragraph{Complexity of the Decoding.}
To perform the distance decoding, for a given pair $u, v \in V$, it suffices to minimize $\cost(u,w) + \cost (v,w)$ over all $w$ belonging to $B'_u(\symb') \cup S'_u$ which are also encoded in the label of $v$. Hence, distance decoding is possible in time $2^{\bigo(\tau)}$. Overall, setting $T:=2^{\bigo(\tau)}$, we obtain the following main result of the section.

\begin{theorem}
There is a $2$-additive distance labeling scheme for general graphs, which achieves decoding time $T$ using labels of size $ \bigo(n \frac{\log^2 \log T}{\log T})$, for any $T \le n$.\qed
\end{theorem}

Finally, we remark on some implications of our result. By a standard argument, converting a $2$-additive approximate distance labeling into an exact one requires an additional label of size $\frac{\log_2 3}{2} n$ bits per node (and an additional $\bigo(\frac{n}{\log n})$ overhead in the space, which is negligible), with each node $u$ encoding the difference between the approximate and real distance value, $\da(u,v)-\cost(u,v)$, for all $v \in \{(u+1) \bmod n, \ldots, (u + \lfloor \frac n 2 \rfloor) \bmod n\}$. The time overhead of the corresponding decoding is $\bigo(1)$. In an analogous manner, converting a $2$-additive approximate distance labeling into an $1$-additive approximate one requires an additional label of size $\frac{1}{2}n$ bits per node. Thus we convert our scheme into an exact distance labeling scheme or $1$-additive scheme achieving $T$ decoding time using labels of size respectively $ \frac{\log_2 3}{2} n + \bigo(n \frac{\log^2 \log T}{\log T})$ or $\frac{1}{2} n + \bigo(n \frac{\log^2 \log T}{\log T})$, for any $T \le n$.

Thus, setting $\tau$ as an arbitrarily small increasing function of $n$, for any desired decoding time $T=\omega(1)$ we can make use of labels of size $o(n)$, $\frac12n + o(n)$ and $\frac{\log_2 3}{2} n + o(n)$ respectively for $2$-additive, $1$-additive and exact distances. Moreover, using this scheme, $\bigo(1)$ decoding time can be achieved for labels of size $\varepsilon  n$, $(\frac{1}{2} + \varepsilon) \cdot n$ and $(\frac{\log_2 3}{2} + \varepsilon) \cdot n$, for any absolute constant $\varepsilon > 0$.

While a slightly stronger in terms of decoding time schemes were presented in  Alstrup \etal~\cite{DBLP:conf/soda/AlstrupGHP16} (achieving $\bigo(1)$ decoding time and labels of size $\frac{\log_2 3}{2} n + o(n)$ and $\frac12n + o(n)$ for exact and $1$-additive distances), we believe that presented here schemes are of independent value due to the simplification of the construction.



\section*{Acknowledgments}
Most of the work was done while PU was affiliated to Aalto University, Finland.
Research partially supported by the National Science Centre, Poland - grant number 2015/17/B/ST6/01897.

\bibliographystyle{abbrv}
\bibliography{biblio}

\end{document}